\documentclass[aps,pra,nofootinbib,twocolumn,floatfix,showpacs,superscriptaddress]{revtex4-1}
\usepackage[utf8]{inputenc}
\usepackage{amsmath,graphicx,epsfig,amsthm,color,bm}
\usepackage{pifont,bbding,amssymb,soul}
\usepackage{indentfirst}
\usepackage{times}
\usepackage{appendix, comment}
\usepackage{multirow}
\usepackage{CJKutf8}
\usepackage{ucs}
\usepackage[encapsulated]{CJK}
\usepackage{tikz,pgfplots}

\newcommand{\bra}[1]{\langle #1|}
\newcommand{\ket}[1]{|#1\rangle}

\newcommand{\braket}[1]{\langle #1 \rangle}

\newcommand{\abs}[1]{\lvert #1\rvert}

\newcommand{\Tr}{\mathrm{Tr}}

\newcommand{\tr}{{\rm tr}}

\newtheorem{theorem}{Theorem}

\usepackage{color}  
\definecolor{shadecolor}{rgb}{0.92,0.92,0.92}  
\usepackage{framed} 
\usepackage{ulem}

\usepackage[colorlinks=true,citecolor=blue,urlcolor=blue]{hyperref}

\usepackage[capitalize]{cleveref}

\begin{document}

\title{{The tightness of multipartite coherence from spectrum estimation}}

\author{Qi-Ming Ding}
%\email{These authors contributed equally to this work}
\affiliation{School of Physics, Shandong University, Jinan 250100, China}

\author{Xiao-Xu Fang}
%\email{These authors contributed equally to this work}
\affiliation{School of Physics, Shandong University, Jinan 250100, China}

\author{He Lu}
\email{luhe@sdu.edu.cn}
\affiliation{School of Physics, Shandong University, Jinan 250100, China}

\begin{abstract}
Detecting multipartite quantum coherence usually requires quantum state reconstruction, which is quite inefficient for large-scale quantum systems. {Along this line of research, several efficient procedures have been proposed to detect multipartite quantum coherence without quantum state reconstruction, among which the spectrum-estimation-based method is suitable for various coherence measures. Here, we first generalize the spectrum-estimation-based method for the geometric measure of coherence. Then, we investigate the tightness of the estimated lower bound of various coherence measures, including the geometric measure of coherence, $l_1$-norm of coherence, the robustness of coherence, and some convex roof quantifiers of coherence multiqubit GHZ states and linear cluster states. Finally, we demonstrate the spectrum-estimation-based method as well as the other two efficient methods by using the same experimental data [Ding \emph{et al.}~Phys. Rev. Research 3, 023228 (2021)]. We observe that the spectrum-estimation-based method outperforms other methods in various coherence measures, which significantly enhances the accuracy of estimation.}
\end{abstract}
\maketitle

\section{Introduction}
\label{sec:sectionI}
	Quantum coherence, as a fundamental characteristic of quantum mechanics, describes the ability of a quantum state to present quantum interference phenomena~\cite{nielsen_quantum_2010}. It also plays a central role in many emerging areas, including quantum metrology~\cite{giovannetti_quantum-enhanced_2004,giovannetti_advances_2011}, nanoscale thermodynamics~\cite{lostaglio_description_2015,narasimhachar_low-temperature_2015,aberg_catalytic_2014,gour_resource_2015}, and energy transportation in the biological system~\cite{huelga_vibrations_2013,lloyd_quantum_2011,lambert_quantum_2013,romero_quantum_2014}. Recently, a rigorous framework for quantifying coherence as a quantum resource was introduced~\cite{baumgratz_quantifying_2014,streltsov_colloquium_2017,hu_quantum_2018}. 	Meanwhile, the framework of resource theory of coherence has been extended from a single party to the multipartite scenario~\cite{bromley_frozen_2015,radhakrishnan_distribution_2016,streltsov_measuring_2015,yao_quantum_2015}. 

	Based on the general framework, several coherence measures have been proposed, such as the $l_1$ norm of coherence, the relative entropy of coherence~\cite{baumgratz_quantifying_2014}, the geometric measure of coherence~\cite{streltsov_measuring_2015}, the robustness of coherence~\cite{napoli_robustness_2016,piani_robustness_2016}, some convex roof quantifiers of coherence~\cite{yuan_intrinsic_2015,winter_operational_2016,zhu_operational_2017,liu_new_2017,qi_measuring_2017}, and others~\cite{shao_fidelity_2015,chin_coherence_2017,rana_trace-distance_2016,zhou_polynomial_2017,xi_coherence_2019,xi_epsilon-smooth_2019,cui_examining_2020}. These coherence measures make it possible to quantify the role of coherence in different quantum information processing tasks, especially in the multipartite scenario, such as quantum state merging~\cite{streltsov_entanglement_2016}, coherence of assistance~\cite{chitambar_assisted_2016}, incoherent teleportation~\cite{streltsov_towards_2017}, coherence localization~\cite{styliaris_quantum_2019}, and anti-noise quantum metrology~\cite{zhang_demonstrating_2019}. However, detecting or estimating most coherence measures requires the reconstruction of quantum states, which is inefficient for large-scale quantum systems.

	Efficient protocols for detecting quantum coherence without quantum state tomography have been recently investigated~\cite{smith_quantifying_2017,wang_directly_2017,yuan_direct_2020,zhang_estimating_2018,yu_detecting_2019,ding_efficient_2021,dai_experimentally_2020,ma_detecting_2021}.
	However, the initial proposals require either complicated experiment settings for multipartite quantum systems~\cite{smith_quantifying_2017,wang_directly_2017,yuan_direct_2020} or complex numerical optimizations~\cite{zhang_estimating_2018}. An experiment-friendly tool, the so-called spectrum-estimation-based method, requires local measurements and simple post-processing~\cite{yu_detecting_2019}, and has been experimentally demonstrated to measure the relative entropy of coherence~\cite{ding_efficient_2021}. Other experiment-friendly tools, such as the fidelity-based estimation method~\cite{dai_experimentally_2020} and the witness-based estimation method~\cite{ma_detecting_2021}, have been successively proposed very recently. {The fidelity-based estimation method delivers lower bounds for coherence concurrence~\cite{qi_measuring_2017}, the geometric measure of coherence~\cite{streltsov_measuring_2015}, and the coherence of formation~\cite{winter_operational_2016}, and the witness-based estimation method can be used to estimate the robustness of coherence~\cite{napoli_robustness_2016}.
	
	Still, there are several unexplored matters along this line of research. First, on the theoretical side, although it has been studied that the spectrum-estimation-based method is capable to detect coherence of several coherence measures~\cite{yu_detecting_2019}, there still exists some coherence measures unexplored. On the experimental side, the realization is focused on the detection of relative entropy of coherence~\cite{ding_efficient_2021}, and its feasibility for other coherence measures has not been tested. Second, the tightness of estimated bounds on multipartite states with spectrum-estimation-based method has not been extensively discussed. Third, while the efficient schemes have been studied either theoretically or experimentally, their feasibility and comparison with the same realistic hardware are under exploration. In particular, implementing efficient measurement schemes and analysing how the noise in realistic hardware affects the measurement accuracy are critical for studying their practical performance with realistic devices.

	The goal of this work is to investigate the spectrum-estimation-based method in three directions: First, we generalize the spectrum-estimation-based method to detect the geometric measure of coherence, which has not been investigated yet. Second, we investigate the tightness of the estimated bound with the spectrum-estimation-based method on multipartite Greenberger-Horne-Zeilinger(GHZ) states and linear cluster states. Finally, we present the comparison of the efficient methods with the same experimental data.
	
	The article is organized as follows. In~\cref{sec:sectionII}, we briefly introduce the theoretical background, including the review of definitions of well-explored coherence measures, the present results of coherence estimation with the spectrum-estimation-based method and the construction of constraint in the spectrum-estimation-based method. In~\cref{sec:sectionIII}, we provide the generalization of the spectrum-estimation-based method for the geometric measure of coherence. In~\cref{sec:sectionVI}, we discuss the tightness of estimated bounds on multipartite states. In~\cref{sec:sectionV}, we present the results of comparison for three estimation methods. Finally, we conclude in~\cref{sec:sectionIV}.}
	
	%%%%%%%%%%%%%%%%%%%%%%%%%%%%%%%%%%%%%%%%%%
\section{Theoretical background}	
\label{sec:sectionII}

\subsection{Review of coherence measures}
A functional $C$ can be regarded as a coherence measure if it satisfies four postulates:~non-negativity, monotonicity, strong monotonicity, and convexity \cite{baumgratz_quantifying_2014}. For a $n$-qubit quantum state $\rho$ in Hilbert space with dimension of $d=2^n$, the relative entropy of coherence $C_{r}(\rho)$, $l_1$ norm of coherence $C_{l_1}(\rho)$~\cite{baumgratz_quantifying_2014} and the geometric measure of coherence $C_{g}(\rho)$~\cite{streltsov_measuring_2015} are distance-based coherence measures, and are defined as 
\begin{equation}
	C_{r}(\rho)=S(\rho_{d})-S(\rho)
	\label{eq:Cr}
\end{equation}
\begin{equation}
	C_{l_1}(\rho)=\sum_{i\neq j}\abs{\rho_{ij}}
	\label{eq:l1}
\end{equation}
\begin{equation}
	C_{g}(\rho)=1-\max_{\sigma \in \mathcal{I}}F(\rho, \sigma)
	\label{eq:Cgg}
\end{equation}
%\begin{align}
%	&C_{l_1}(\rho)=\sum_{i\neq j}\abs{\rho_{ij}}\\
%	&C_{r}(\rho)=S(\rho_{d})-S(\rho),\\
%	&C_{g}(\rho)=1-\max_{\sigma \in \mathcal{I}}F(\rho, \sigma)
%\end{align}	
respectively, where $S=-\tr[\rho\log_2\rho]$ is the von Neumann entropy, $\rho_{d}$ is the diagonal part of $ \rho $ in the incoherent basis and $F(\rho, \sigma)=\|\sqrt{\varrho} \sqrt{\sigma}\|_{1}^{2}$. 

%Let $\mathcal{D}\left(\mathbb{C}^{d}\right)$ be the convex set of density operators acting on a $d$ -dimensional Hilbert space, and let $\mathcal{I} \subset \mathcal{D}\left(\mathbb{C}^{d}\right)$ be the subset of incoherent states. We define the robustness of coherence (ROC) of a quantum state $\rho \in \mathcal{D}\left(\mathbb{C}^{d}\right)$ as

The robustness of coherence is defined as,
\begin{equation}
	C_{R}(\rho)= \min _{\tau}\left\{s \geq 0 \mid \frac{\rho+s \tau}{1+s}=: \delta \in \mathcal{I}\right\},
	\label{eq:CR_original}
\end{equation} 
where $\mathcal{I}$ is the set of incoherent states and $C_{R}(\rho)$ denotes the minimum weight of another state $\tau$ such that its convex mixture with $\rho$ yields an incoherent state $\delta$~\cite{napoli_robustness_2016}.

%uhlmann_roofs_2010,
Another kind of coherence measure is based on convex roof construction~\cite{yuan_intrinsic_2015,zhu_operational_2017}, such as coherence concurrence $\tilde{C}_{l_{1}}$~\cite{qi_measuring_2017}, and coherence of formation $C_f$~\cite{winter_operational_2016} in form of 
\begin{align}
	\tilde{C}_{l_{1}}(\rho)&=\inf_{\{p_i,\ket{\varphi_i}\}}\sum_ip_iC_{l_{1}}(\ket{\varphi_i}),\label{eq:Cl1_cf}\\
	C_f(\rho)&=\inf_{\{p_i,\ket{\varphi_i}\}}\sum_ip_iC_r(\ket{\varphi_i}),\label{eq:Cf_cf}
\end{align}
where the infimum is taken over all pure state decomposition of $\rho=\sum_{i} p_{i} \ket{\psi_{i}}\bra{\psi_{i}}$. 

It is also important to consider the $l_2$ norm of coherence $C_{l_{2}}(\rho) =\min_{\delta \in \mathcal{I}}\abs{\abs{\rho-\delta}}_{l_2}^{2}= \sum_{i \neq j}\left|\rho_{ij}\right|^{2}=S_{L}(\boldsymbol{d})-S_{L}(\boldsymbol{\lambda})$ with $S_{L}(\boldsymbol{p})=1-\sum_{i=1}^{d} p_{i}^{2}$ being the Tsallis-2 entropy or linear entropy, and $\boldsymbol{\lambda}$ is the spectrum of $\rho$~\cite{baumgratz_quantifying_2014}.

The different coherence measure plays different roles in quantum information processing. The relative entropy of coherence plays a crucial role in coherence distillation~\cite{winter_operational_2016}, coherence freezing~\cite{bromley_frozen_2015,yu_measure-independent_2016}, and the secrete key rate in quantum key distribution~\cite{ma_operational_2019}. The $l_1$-norm of coherence is closely related to quantum multi-slit interference experiments~\cite{bera_duality_2015} and is used to explore the superiority of quantum algorithms~\cite{hillery_coherence_2016,shi_coherence_2017,liu_coherence_2019}. The robustness of coherence has a direct connection with the success probability in quantum discrimination tasks~\cite{napoli_robustness_2016,piani_robustness_2016,takagi_operational_2019}. The coherence of formation represents the coherence cost, i.e., the minimum rate of a maximally coherent pure state consumed to prepare the given state under incoherent and strictly incoherent operations~\cite{winter_operational_2016}. The coherence concurrence~\cite{qi_measuring_2017} and the geometric measure of coherence~\cite{streltsov_measuring_2015} can be used to investigate the relationship between the resource theory of coherence and entanglement.

\subsection{Spectrum-estimation-based method for coherence detection}
We consider the relative entropy of coherence $C_r(\rho)$ and $l_2$ norm of coherence $C_{l_2}(\rho)$ that can be estimated with spectrum-estimation-based algorithm~\cite{yu_detecting_2019}. The former is
\begin{eqnarray}
	C_{r}(\rho) \geq l_{C_r}(\rho)= S_{\text{VN}}{(\boldsymbol{d})}-S_{\text{VN}}\left(\boldsymbol{d} \vee\left(\wedge_{\boldsymbol{p} \in X} \boldsymbol{p}\right)\right),
	\label{eq:MultpartiteEstimationCr}
\end{eqnarray}
where $S_{\text{VN}}=-\tr[\rho \log \rho]$ being the von Neumann entropy. The latter is determined by
\begin{equation}\label{eq:MultpartiteEstimationCl2}
	C_{l_{2}}(\rho)\geq l_{C_{l_2}}(\rho)= S_{L}(\boldsymbol{d})-S_{L}\left(\boldsymbol{d} \vee \left(\wedge_{\boldsymbol{p} \in X} \boldsymbol{p}\right)\right).
\end{equation}
$\boldsymbol{d}=\left(d_{1}, \ldots, d_{2^{n}}\right)$ are the diagonal elements of $\rho$, $\boldsymbol{p}=\left(p_{1}, \ldots, p_{2^{n}}\right)$ is the estimated probability distribution of the measurement on a certain entangled basis $\left\{\left|\psi_{k}\right\rangle\right\}_{k=1}^{2^{n}}, \vee$ is majorization joint, and $\wedge_{\boldsymbol{p} \in X} \boldsymbol{p}$ is the majorization meet of all probability distributions in $X$~\cite{yu_detecting_2019}. Here the majorization join and meet are defined based on majorization. Without loss of generality, the probability distribution $\boldsymbol{p}$ in $X$ set can be restricted by some equality constraints and inequality constraints, i.e., $X=\{\bm{p}| A\bm{p}\geq\bm{\alpha}, B\bm{p}=\bm{\beta}\}$.

$C_{l_{1}}(\rho)$ and $C_{R}(\rho)$ have relations to $l_{C_{l_2}}(\rho)$ as 
\begin{eqnarray}
	\begin{aligned}
		C_{l_{1}}(\rho) \geq l_{C_{l_1}}(\rho) &=\sqrt{2 l_{C_{l_2}}(\rho)} \sum_{k=1}^{d(d-1) / 2} \sqrt{\hat{v}_{k}}, \\
		C_{R}(\rho) \geq l_{C_{R}}(\rho) &=\sqrt{2 l_{C_{l_2}}(\rho)} \sum_{k=1}^{d(d-1) / 2} \frac{\hat{v}_{k}}{\sqrt{u_{k}}}.
	\end{aligned}
	\label{eq:cl1cR}
\end{eqnarray}
where $\boldsymbol{u}^{\downarrow}= \left(u_{k}\right)_{k=1}^{d(d-1) / 2} $ is a descending sequence with $u_{k}=(2 d_{i} d_{j} / C_{l_2}(\rho))_{i<j}$, 
\begin{eqnarray}
	\hat{v}_{k}=\begin{cases}
		u_{k} & \text { for }  k \leq M \\
		1-\sum_{l=1}^{M} u_{l} & \text { for }  k=M+1 \\
		0 & \text { for }  k>M+1
	\end{cases},
	\label{eq:vandk}
\end{eqnarray}
and $M$ is the largest integer satisfying $\sum_{l=1}^{M} u_{l} \leq 1$. It is notable to consider the following case: if $u_{1} \geq 1$, then $\hat{v}_{1}=1$ and $\hat{v}_{k}=0$ for all $ k \neq 1 $ {according to Eq.~(\ref{eq:vandk}), which leads $\hat{v}_{k}=\hat{u}_{k}=(1,0,\cdots,0)$}.

The convex roof coherence measures $C_{f}(\rho)$ and $\tilde{C_{l_1}}(\rho)$ have relations to $l_{C_r}(\rho)$ and $l_{C_{l_1}}(\rho)$, respectively. It is well known that the value of convex roof coherence measure is greater than that of distance-based coherence measure, so that it is natural to obtain
\begin{equation}
	\begin{aligned}
		C_{f}(\rho) &\geq C_{r}(\rho) \geq l_{C_r}(\rho),\\
		\tilde{C_{l_1}}(\rho) &\geq C_{l_{1}}(\rho) \geq l_{C_{l_1}}(\rho).
	\end{aligned}
	\label{eq:convexr}
\end{equation}
{Henceforth, we denote $l_{C}(\cdot)$ as results from estimations, while $C(\cdot)$ as the results calculated with density matrix (theory) or reconstructed $\rho_{\text{expt}}^{\psi}$ (experiment).}

\subsection{Constructing constraint with stabilizer theory}
{For a $n$-qubit stabilizer state $\ket{\psi_k}$, the constraint $X=\{\bm{p}| A\bm{p}\geq\bm{\alpha}, B\bm{p}=\bm{\beta}\}$ can be constructed by the stabilizer $S_i$ of $\ket{\psi_k}$. However, considering the experimental imperfections, the constraint can be relaxed as~\cite{ding_efficient_2021}  
	\begin{equation}
		A=\mathbb{I}_{d}, \alpha=0,
	\end{equation}
	and 
	\begin{equation}
		\begin{bmatrix}
			\braket{S_1}-w\sigma_1 \\ \vdots \\ \braket{S_{d}}-w\sigma_{d}
		\end{bmatrix}\leq
		B\cdot\bm{p}
		\leq\begin{bmatrix}
			\braket{S_1}+w\sigma_1 \\ \vdots \\ \braket{S_{d}}+w\sigma_{d}
		\end{bmatrix},
	\end{equation}
	where $\sigma_i$ is the statistical error associated with experimentally obtained $\{\braket{S_i}\}$, and $w\sigma_i$ with $w\geq0$ is the deviation to the mean value $\braket{S_i}$ represented in $\sigma_i$. Note that $\braket{\mathbb{I}^{\otimes n}}=1$ must be set in the constraint.}

\section{Detecting the geometric measure of coherence with spectrum-estimation-based method}
\label{sec:sectionIII}
We present that the geometric measure of coherence $C_{g}(\rho)$ is related to $l_{C_{l_2}}(\rho)$.
\begin{theorem}
	The lower bound of the geometric measure of coherence $l_{C_{g}}(\rho)$ of a $n$-qubit quantum state is related to $l_{C_{l_2}}(\rho)$ by 
	\begin{equation}\label{eq:MultpartiteEstimationCg}
		C_{g}(\rho) \geq l_{C_{g}}(\rho) = \frac{d-1}{d}\left(1-\sqrt{1-\frac{d}{d-1} l_{C_{l_2}}(\rho)} \right).
	\end{equation}
\end{theorem}

\begin{proof}
	It has been proved that~\cite{zhang_estimation_2017}
	\begin{equation}
	\label{eq:Cg}
		C_{g}(\rho) \geq 1-\frac{1}{d}-\frac{d-1}{d} \sqrt{1-\frac{d}{d-1}\left(\operatorname{Tr}\left(\rho^{2}\right)-\sum_{i=1}^{d} \rho_{i i}^{2}\right)}.
	\end{equation}
	We rewrite the right side of equation (\ref{eq:Cg}) and denote the function of $C_{l_{2}}$ as $G(C_{l_{2}})$ by
	\begin{equation}
		\begin{split}
			&1-\frac{1}{d}-\frac{d-1}{d} \sqrt{1-\frac{d}{d-1}\left(\operatorname{Tr}\left(\rho^{2}\right)-\sum_{i=1}^{d} \rho_{i i}^{2}\right)}\\
			&= \frac{d-1}{d}\left(1-\sqrt{1-\frac{d}{d-1}{\left(S_{L}(\boldsymbol{d})-S_{L}(\boldsymbol{\lambda})\right)} }\right)\\
			&= \frac{d-1}{d}\left(1-\sqrt{1-\frac{d}{d-1}{\left(C_{l_{2}}(\rho)\right)} }\right)\\
			&= G\left(d,C_{l_{2}}(\rho)\right).
		\end{split}
	\end{equation}
	It is easy to check that $G\left(d,C_{l_{2}}(\rho)\right)$ is an increasing function of $C_{l_{2}} (\rho)$ when $d>1$, which implies $C_{g}(\rho) \geq G\left(d,C_{l_{2}}(\rho)\right) \geq G\left(d,l_{C_{l_{2}}}(\rho)\right)= l_{C_{g}}(\rho)$, i.e.,	
	\begin{equation}
		C_{g}(\rho) \geq l_{C_{g}}(\rho) = \frac{d-1}{d}\left(1-\sqrt{1-\frac{d}{d-1} l_{C_{l_2}}(\rho)} \right).
			\label{eq:CGGG}
	\end{equation}	  
\end{proof}

\section{Tightness of estimated lower bounds}\label{sec:sectionVI}
{The lower bounds $l_{C_{r}}$, $l_{C_{l_2}}$, $l_{C_{l_1}}$ and $l_{C_{R}}$ are tight for pure states~\cite{yu_detecting_2019}. $l_{\tilde{C_{l_1}}}$ and $l_{C_{f}}$ related to $l_{C_{l_1}}$ and $l_{C_{r}}$ as shown in~\cref{eq:convexr} are tight for pure states as well. However, the tightness of $l_{C_{g}}$ is quite different. As shown in~\cref{eq:MultpartiteEstimationCg}, $l_{C_{g}}$ is related to the dimension of quantum system $d$ as well as $l_{C_{l_{2}}}$. Although $l_{C_{l_{2}}}$ is tight for stabilizer states, $l_{C_{g}}$ is generally not due to the fact of $\frac{d-1}{d}$. The equality in~\cref{eq:MultpartiteEstimationCg} holds for a special family of states, i.e., the maximal coherent states $\ket{\Psi_{d}}=\frac{1}{\sqrt{d}} \sum_{\alpha=0}^{d-1} e^{i\theta_{\alpha}}|\alpha\rangle$~\cite{zhang_estimation_2017}.

To investigate the tightness of the estimated bounds $l_C$ on multipartite states, we consider the graph states
\begin{equation}
\ket{G}=\prod_{(i,j)\in E}CZ_{(i,j)}\ket{+}^{\otimes n}.\label{G}
\end{equation}
For a target graph $G$ with $n$ qubits (vertices), the initial states are the tensor product of $\ket{+}=(\ket{0}+\ket{1})/\sqrt{2}$. An edge $(i,j)\in E$ corresponds to a two-qubit controlled Z gate $CZ_{(i,j)}$ acting on two qubits $i$ and $j$. Particularly, we investigate two types of graphs. The first one is star graph, and the corresponding state is $n$-qubit GHZ states $\ket{\text{GHZ}_n}=\frac{1}{\sqrt{2}}(\ket{0}^{\otimes n}+\ket{1}^{\otimes n})$ with local unitary transformations acting on one or more of the qubits. The second one is linear graph, and the corresponding state is $n$-qubit linear cluster state $\ket{\text{C}_n}$~\cite{Briegel2001}, which is the ground state of the Hamiltonian 
\begin{equation}
    \mathcal H(n)=\sum_{i=2}^{n-1}Z^{(i-1)}X^{(i)}Z^{(i-1)}-X^{(1)}Z^{(2)}-Z^{(n-1)}X^{(n)},
\end{equation}
where $X^{(i)}$, $Y^{(i)}$ and $Z^{(i)}$ denote the Pauli matrices acting on qubit $i$.
\begin{figure*}[ht!bp]%4.25[width=0.95\columnwidth]
	\includegraphics[scale=0.8]{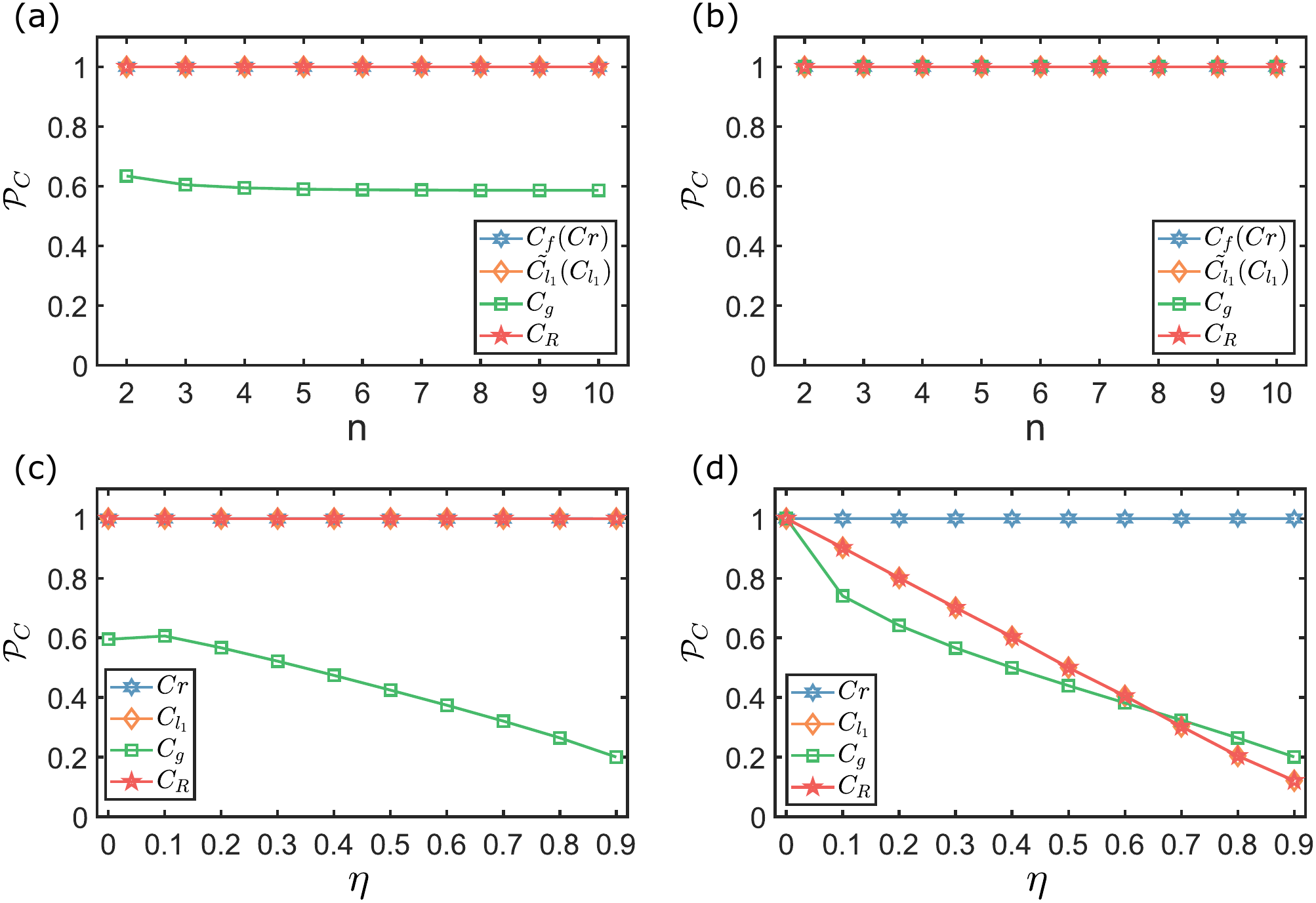} 
	\caption{Theoretical results $\mathcal P_C$ of coherence measures of $ C_f(C_r), \tilde{C_{l_1}}(C_{l_1}), C_g$ and $C_R$ on state (a), $\ket{\text{GHZ}_n}$, (b), $\ket{\text{C}_n}$, (c), $\rho_{\text{Noisy}}^{\text{GHZ}_4}$, (d), $\rho_{\text{Noisy}}^{\text{C}_4}$.}
	\label{Fig:GHZ}
\end{figure*}

For $\ket{\text{GHZ}_n}$ and $\ket{\text{C}_n}$ with $n$ up to 10, we calculate $\mathcal P_C=l_C/C$ to indicate the tightness (accuracy) of estimations, and the results are shown in~\cref{Fig:GHZ}(a) and~\cref{Fig:GHZ}(b), respectively. For $\ket{\text{GHZ}_n}$, we observe that $\mathcal P_C$ is 1 in the estimation of $C_f(C_r),\tilde{C_{l_1}}(C_{l_1})$ and $C_R$, which indicates the corresponding bounds are tight as the target states are pure state. The reason is that $l_{C_{g}}$ is determined by $d$ and $l_{C_{l_{2}}}$ as shown in~\cref{eq:MultpartiteEstimationCg}. For $\ket{\text{GHZ}_n}$, $l_{C_{l_{2}}}=1/2$ regardless of $n$. Then, we take the partial derivative of $l_{C_{g}}$ in~\cref{eq:MultpartiteEstimationCg} with respect to $d$, and obtain
\begin{eqnarray}
	\begin{aligned}
			\frac{\partial l_{C_{g}}}{\partial d} &= \frac{-\left( 2\left(d-1\right)[\sqrt{1-\frac{d}{d-1}l_{C_{l_{2}}}}-1]+dl_{C_{l_{2}}}\right)}{\sqrt{1-\frac{d}{d-1}l_{C_{l_{2}}}}\left(2d^2-2d^3\right)}\\
			&\leq -\frac{3(d-1)-2(d-1)\sqrt{1-\frac{d}{d-1}l_{C_{l_{2}}}}}{\sqrt{1-\frac{d}{d-1}l_{C_{l_{2}}}}\left(2d^3-2d^2\right)}\\
			& \leq 0.
	\end{aligned}
\end{eqnarray}
It is clear that $l_{C_{g}}$ is monotonically decreasing with respect to $d$. Note that $l_{C_{g}}\to (\sqrt{2}-1)/\sqrt{2}\approx 0.2929$ and $\mathcal P_C\to 2(\sqrt{2}-1)/\sqrt{2}\approx 0.5858 $ when $d\to\infty$. 

The results of $\ket{\text{C}_n}$ is quite different as shown in~\cref{Fig:GHZ}(b). $\mathcal P_C$ is 1 in the estimation of $C_g$ as $\ket{\text{C}_n}$ is in the form of maximally coherent state $\ket{\Psi_{d}}=\frac{1}{\sqrt{d}} \sum_{\alpha=0}^{d-1} e^{i\theta_{\alpha}}\ket{\alpha}$. For example, $\ket{\text{C}_3}=(\ket{+0+}+\ket{-1-})/\sqrt{2}$ and we rewrite it in the computational basis
\begin{equation}
\begin{split}
     \ket{\text{C}_3}=\frac{1}{\sqrt{2^3}}(&\ket{000}+\ket{001}+\ket{010}-\ket{011}\\
     &+\ket{100}+\ket{101}-\ket{110}+\ket{111}).    
\end{split}
\end{equation}
By re-encoding $\ket{\alpha_1\alpha_2\alpha_3}$ to $\ket{\alpha}$ by $\alpha=\alpha_12^2+\alpha_22^1+\alpha_32^0$, $\ket{\text{C}_3}$ can be represented in the form of maximally coherent state $\frac{1}{\sqrt{d}} \sum_{\alpha=0}^{d-1} e^{i\theta_{\alpha}}\ket{\alpha}$ with $\bm{\theta_\alpha}=(0,0,0,\pi,0,0,\pi,0)$.

Furthermore, we investigate the robustness of $\mathcal P$ of GHZ states and linear cluster states in a noisy environment. We consider the following imperfect GHZ state and linear cluster state
\begin{equation}
\rho_{\text{Noisy}}^{\psi}= (1-\eta) \ket{\psi}\bra{\psi} + \frac{\eta}{d}\mathbb{I}_{d},
\label{eq:noisystate}
\end{equation}
with $\ket{\psi}$ being either $\ket{\text{GHZ}_n}$ or $\ket{\text{C}_n}$ and $0\leq\eta\leq1$. Note that $\rho_{\text{Noisy}}^{\psi}$ can be written in the form of graph-diagonal state , i.e., $\rho_{\text{Noisy}}^{\psi}=\sum\lambda_k\ket{\psi_k}\bra{\psi_k}$ so that $l_{C_r}$ and $l_{C_{l_2}}$ are tight for $\rho_{\text{Noisy}}^{\psi}$~\cite{ding_efficient_2021}. The estimation of $l_{C_{l_1}}$ is equivalent to the optimization of 
\begin{equation}
	\begin{split}
		\underset{v_k}{\text{minimize}} & \sqrt{2l_{C_{l_2}}} \sum_{k=1}^{d(d-1)/2} \sqrt{v_k} \\
		\text {subject to} & \sum_{k=1}^{d(d-1) / 2} v_{k}=1, \\
		& 0 \leqslant v_{k} \leqslant u_{k}.
	\end{split}
	\label{eq:CR_check_1}
\end{equation}

For $\rho_{\text{Noisy}}^{\text{GHZ}_n}$ in form of
\begin{equation}
    \rho_{\text{Noisy}}^{\text{GHZ}_n}=  \begin{pmatrix}
     \frac{1}{2}(1-\eta)+\frac{1}{2^n}\eta & 0 & \dots & \frac{1}{2}(1-\eta) \\
    0 & \frac{1}{2^n}\eta & \dots & 0 \\
    \vdots & \vdots & \ddots & \vdots \\
    0 &  \dots &\frac{1}{2^n}\eta& 0 \\
    \frac{1}{2}(1-\eta) & 0 & \dots & \frac{1}{2}(1-\eta)+\frac{1}{2^n}\eta
  \end{pmatrix},
\end{equation}
it is easy to calculate that $u_{1} \geq 1$ and $\hat{v}_{k}=\hat{u}_{k}=(1,0,\cdots,0)$. As $l_{C_{l_2}}$ is tight for $\rho_{\text{Noisy}}^{\text{GHZ}_n}$ so that we can obtain
\begin{equation}
\begin{split}
    l_{C_{l_1}}&=\sqrt{2C_{l_2}}=\sqrt{2\sum_{i\neq j}\abs{\rho_{\text{Noisy}}^{\text{GHZ}_n}}_{ij}^2}\\
    &=1-\eta=\sum_{i\neq j}\abs{\rho_{\text{Noisy}}^{\text{GHZ}_n}}_{ij}=C_{l_1},
\end{split}
\end{equation}
which indicates $l_{C_{l_1}}$ is tight for $\rho_{\text{Noisy}}^{\text{GHZ}_n}$. Following the same way, We can also obtain $l_{C_R}=C_R$. 

For $\rho_{\text{Noisy}}^{\text{C}_n}$, as its matrix elemetns satisfy 
\begin{eqnarray}
	\abs{\rho_{\text{Noisy}}^{\text{C}_n}}_{ij}=\begin{cases}
		\frac{1}{d} & \text { for }  i=j \\
		\frac{1-\eta}{d}  & \text { for }  i\neq j
	\end{cases},
	\label{eq:C4elements}
\end{eqnarray}
so that we can calculate the $C_{l_{2}}=l_{C_{l_{2}}}= \frac{d-1}{d}(1-\eta)^2$, $C_R=C_{l_1}=(d-1)(1-\eta)$ and 
\begin{equation}
\begin{aligned}
   \hat{v}_{k}= &\hat{u}_{k}= \{\underbrace{\frac{1}{\frac{d(d-1)}{2}(1-\eta)^2},...,\frac{1}{\frac{d(d-1)}{2}(1-\eta)^2}}_{M}\\
   &,1-\frac{M}{\frac{d(d-1)}{2}(1-\eta)^2},\underbrace{0,0,0,...,0,0}_{\frac{d(d-1)}{2}-(M+1)}\}.
\end{aligned}
\end{equation}
$M=\lfloor \frac{d(d-1)}{2}(1-\eta)^2 \rfloor$ is the largest integer satisfying $\sum_{l=1}^M u_{l}  \leq 1$. With $l_{C_{l_{2}}}$ and $M$, we can calculate $l_{C_{l_1}}$ and $l_{C_R}$
\begin{equation}
	\begin{split}
			l_{C_{l_1}} = l_{C_R} &= \sqrt{\frac{2(d-1)}{d}(1-\eta)^2} \left({M \sqrt{\frac{1}{\frac{d(d-1)}{2}(1-\eta)^2}}} \right. \\ &+ \left. { \sqrt{1-\frac{M}{\frac{d(d-1)}{2}(1-\eta)^2}}} \right).
	\end{split}
	\label{eq:noisyC4}
\end{equation}
As $M\approx \frac{d(d-1)}{2}(1-\eta)^2$ so we have $l_{C_{l_1}} = l_{C_R}\approx (d-1)(1-\eta)^2$. Therefore, $\mathcal P_C$ of $l_{C_{l_1}}$ and $l_{C_R}$ for noisy cluster state is $l_{C_{l_1}}/C_{l_1}=l_{C_R}/C_R\approx1-\eta$. 

To give an intuitive illustration of our conclusion about tightness of $l_C$ on noisy states, we calculate $\mathcal P_C$ on 4-qubit noisy GHZ state and linear cluster state, i.e., $\ket{\text{GHZ}_4}=(\ket{0000}+\ket{1111})/\sqrt{2}$ and $\ket{\text{C}_4}=(\ket{+0+0}+\ket{+0-1}+\ket{-1-0}+\ket{-1+1})/
2$. The results are shown in~\cref{Fig:GHZ}(c) and~\cref{Fig:GHZ}(d), respectively. In~\cref{Fig:GHZ}(c), $\mathcal P_C$ of $l_{C_r}$, $l_{C_{l_1}}$ and $l_{C_R}$ are still tight. In~\cref{Fig:GHZ}(d), $\mathcal P_C$ of $l_{C_r}$ is tight while $\mathcal P_C$ of $l_{C_{l_1}}$ and $l_{C_R}$ linearly decrease with $\eta$. $l_{C_g}$ also exhibits linear decrease with $\eta$ in~\cref{Fig:GHZ}(c) and~\cref{Fig:GHZ}(d) because $l_{C_{l_2}}\sim(1-\eta)^2$.}

\section{Comparison with other coherence estimation methods}
\label{sec:sectionV}

	\begin{table*}[ht!bp]%The best place to locate the table environment is directly after its first reference in text
	\caption{\label{tab:varimethods}%
		Comparison of the spectrum-estimation-based~\cite{yu_detecting_2019}, fidelity-based~\cite{dai_experimentally_2020} and witness-based coherence estimation methods~\cite{ma_detecting_2021} on $\rho_{\text{expt}}^{\text{GHZ}_3}$ and $\rho_{\text{expt}}^{\text{GHZ}_4}$. The cases of $W_1$ and $W_3$ are discussed in Ref.~\cite{ma_detecting_2021}.
	}
	\begin{ruledtabular}
		\begin{tabular}{cccccc}
			Coherence Measure 	&    Method   & \multicolumn{2}{c}{$\rho_{\text{expt}}^{\text{GHZ}_3}$} & \multicolumn{2}{c}{$\rho_{\text{expt}}^{\text{GHZ}_4}$} \\
			&       &  $l_{C}^{\text{max}}$ & $\mathcal{P}_{C}$  & $l_{C}^{\text{max}}$ & $\mathcal{P}_{C}$ \\ 	\colrule
			\multirow{3}[0]{*}{$Cr/C_f $} &  Tomography & 0.8755(19) &  & 0.9059(29) &  \\
			& Spectrum Est. & 0.8099& 92.51(22)\% & 0.8680 & 95.81(32)\% \\
			& Fid.-Based Est. & 0.2216(2) & 25.31(31)\% & 0.2163(3) & 34.91(46)\% \\ 	
			\multirow{3}[0]{*}{$C_{l_1}/\tilde{C_{l_1}} $} & Tomography &  1.2810(47) &  &  1.4248(46) & \\
			& Spectrum Est. & 0.9393 & 73.09(37)\% & 0.9420 & 66.11(32)\% \\
			& Fid.-Based Est. & 0.9287(6) & 72.50(43)\% & 0.9139(8) & 64.14(41)\% \\ 	
			\multirow{3}[0]{*}{$C_g $} &  Tomography & { 0.3571(11)} &  & { 0.3728(17)} &  \\
			& Spectrum Est. & 0.2789 & 78.10(31)\% & 0.2710 & 72.69(46)\% \\
			& Fid.-Based Est. & 0.0229(0) & 6.41(31)\% & 0.0222(0) & 5.95(46)\% \\ 	
			\multirow{3}[1]{*}{$C_R $} & { Tomography} & {1.2680(50)} &  & { 1.3942(48)} &    \\
			& Spectrum Est. & 0.9393 & 73.84(39)\% & 0.9420 & 67.56(34)\% \\
			& $-\Tr(W_3\rho)$\footnote{$W_3=\frac{1}{2} \mathbb{I} - \ket{\text{GHZ}_n}\bra{\text{GHZ}_n}$ } & 0.4644(3) & 36.62(46)\% & 0.4659(4) & 33.42(43)\% \\
			& $-\Tr(W_1\rho)$\footnote{$W_1=\Delta(\ket{\text{GHZ}_n}\bra{\text{GHZ}_n})-\ket{\text{GHZ}_n}\bra{\text{GHZ}_n}$, where $\Delta(\rho)=\sum_{i=1}^{d}\ket{i}\bra{i}\rho\ket{i}\bra{i}$}  & 0.4714(3) & 37.17(46)\% & 0.4684(4) & 33.60(43)\% \\
		\end{tabular}
	\end{ruledtabular}
\end{table*}
{Besides the spectrum-estimation-based method, another two efficient coherence estimation methods for multipartite states have been proposed recently, namely the fidelity-based estimation method~\cite{dai_experimentally_2020} and the witness-based estimation method~\cite{ma_detecting_2021}, respectively. Specifically, $C_{f}$,$C_{g}$,$\tilde{C_{l_{1}}}$ can be estimated via the fidelity-based estimation method and $C_{R}$ can be estimated via the witness-based estimation method. In this section, we compare the accuracy of $l_C$ with difference estimation method with experimental data of  $\rho_{\text{expt}}^{\text{GHZ}_3}$ and $\rho_{\text{expt}}^{\text{GHZ}_4}$ from Ref.~\cite{ding_efficient_2021}.

To this end, we first estimate $l_C$ of the coherence measures $C$ introduced in~\cref{sec:sectionII} on states $\rho_{\text{expt}}^{\text{GHZ}_3}$ and $\rho_{\text{expt}}^{\text{GHZ}_4}$ via spectrum-estimation-based method. We employ the experimentally obtained expected values of the stabilizing operators $\mathcal{S}^{\text{GHZ}_n}$ and the corresponding statistical errors $\sigma$ to construct constraints in $X$. We denote the lower bound of estimated multipartite coherence as $l_{C,m}^{w}$, where $C$ is the coherence measure $\in\{ C_f(C_r), \tilde{C_{l_1}}(C_{l_1}), C_g,C_R\}$ and $m\leq2^{n}-1$ is the stabilizing operators we selected for construction of constraints in $X$. In our estimations, all results are obtained by setting $w=3$. Here, we only consider the case  of maximal $l_{C,m}$. In the ideal case, the maximal $l_{C,m}$ is obtained by setting all $m$ stabilizing operators in the constraint. However, a larger $m$ might lead to the case of no feasible solution due to the imperfections in experiments. In practice, the maximal estimated coherence is often obtained with $m\leq 2^{n}-1$ stabilizing operators. Let $l_{C}^{\text{max}}$ be the maximal estimated coherence over all subsets $\{S_i\}$, where the number of subset is $\sum_{m=1}^{2^n-1} \binom{2^n - 1}{m}= 2^{2^{n}-1}-1$. The results of $l_{C}^{\text{max}}$ are shown in~\cref{tab:varimethods}. 

The accuracy estimated bounds is indicated by $\mathcal P_C=l_{C}^{\text{max}}/C$. Note that $C_{r}$ and $C_{l_{1}}$ of $\rho^{\psi}_{\text{expt}}$ can be calculated directly according to the definition in~\cref{eq:Cr} and~\cref{eq:l1}, while the calculations of $C_{g}(\cdot)$ and $C_{R}(\cdot)$ require converting them to the convex optimization problem \cite{napoli_robustness_2016,piani_robustness_2016,zhang_numerical_2020} and the corresponding solution~\cite{boyd_convex_2004,cvx,gb08}. The calculation of ${C_{f}},\tilde{C}_{l_{1}}$ requires optimizing all pure state decomposition, and there is no general method for analytical and numerical solutions except a few special cases. Therefore, we replace ${C_{f}},\tilde{C}_{l_{1}}$ of these tomographic states by their $C_{r},C_{l_{1}}$ when calculating the estimated accuracy, respectively. The replacement increases $\mathcal{P}_{C}$ when we compare the two estimation methods of spectrum-estimation-based and fidelity-based so that it does not affect our conclusion about the comparison. 

We also perform the fidelity-based estimation method and witness-based estimation method on the same experimental data to obtain $l_C$ and $\mathcal{P}_C$. The results of $l_C$ and $\mathcal{P}_C$ with these three estimation methods are shown in~\cref{tab:varimethods}. We find that the spectrum-estimation-based and fidelity-based coherence estimation methods have similar performance on the estimation of $\tilde{C_{l_{1}}}$, in which the accuracy is beyond 0.7 for $\rho^{\text{GHZ}_3}_{\text{expt}}$ and 0.6 for $\rho^{\text{GHZ}_4}_{\text{expt}}$. Importantly, the spectrum-estimation-based method shows a significant enhancement in the estimation of $C_{f}$ and $C_{g}$ compared with the fidelity-based method, as well as in the estimation of $C_{R}$ compared with the witness-based method. }

\section{Conclusions}
\label{sec:sectionIV}

In this work, we first develop the approach to estimating the lower bound of coherence for the geometric measure of coherence via the spectrum-estimation-based method, {i.e., we present the relation between the geometric measure of coherence and $l_2$ norm of coherence. Then, we investigate the tightness of estimations of various coherence measures on GHZ states and linear cluster states, including the geometric measure of coherence, the relative entropy of coherence, the $l_1$-norm of coherence, the robustness of coherence, and some convex roof quantifiers of coherence. Finally, we compare the accuracy of the estimated lower bound with the spectrum-estimation-based method, fidelity-based estimation method, and the witness-based estimation method on the same experimental data. 

We conclude that the spectrum-estimation-based method is an efficient toolbox to indicate various multipartite coherence measures. For $n$-qubit stabilizer states, it only requires at most $n$ measurements instead of $3^n$ measurements required in quantum state tomography. Second, the tightness of the lower bound is not only determined by whether the target state is pure or mixed but also by the coherence measures. We give examples that the lower bound of the geometric measure of coherence is tight for $n$-qubit linear cluster states but is not tight for noisy $n$-qubit GHZ states, and the lower bounds of the robustness of coherence and $l_1$-norm of coherence are tight for noisy $n$-qubit GHZ states but is not tight for noisy $n$-qubit linear cluster states. Third, we find that the spectrum-estimation-based method has a significant improvement in coherence estimation compared to fidelity-based method and the witness-based method.}

\begin{acknowledgments}
We are grateful to anonymous referees for providing very useful comments on earlier versions of this manuscript. This work is supported by the National Natural Science Foundation of China (Grant No. 11974213 and No. 92065112), National Key R\&D Program of China (Grant No. 2019YFA0308200), and Shandong Provincial Natural Science Foundation (Grant No. ZR2019MA001 and No. ZR2020JQ05), Taishan Scholar of Shandong Province (Grant No. tsqn202103013) and Shandong University Multidisciplinary Research and Innovation Team of Young Scholars (Grant No. 2020QNQT).
\end{acknowledgments}

\bibliography{COC.bib}

\end{document}